\documentclass[11pt]{article}

\usepackage{ams}
\usepackage{amssymb}
\usepackage{graphicx}
\usepackage{caption}
\usepackage{pgf}
\usepackage{tikz}
\usepackage{cite} 
\usepackage{amsthm}

\usetikzlibrary{arrows,automata}

\usepackage{latexsym}
\usepackage{amsfonts}
\usepackage[english]{babel}
\usepackage{xy}
\usepackage{fancyhdr}
\usepackage[a4paper,top=3cm,bottom=3.5cm,left=2.3cm,right=2.3cm,bindingoffset=5mm]{geometry}

\usepackage{setspace}

\newtheorem{lemma}{Lemma}
\newtheorem{corollary}{Corollary}
\newtheorem{proposition}{Proposition}
\newtheorem{example}{Example}

\newcommand{\todo}[1]{{\bf  *** #1 ***}}

\newcommand{\DN}{\mbox{\boldmath $ D $}_N}
\newcommand{\sample}{\mbox{\boldmath $ s $}}

\newcommand{\YN}{\mbox{\boldmath $ Y $}_N}
\newcommand{\yinfty}{\mbox{\boldmath $ y $}_{\infty}}
\newcommand{\xinfty}{\mbox{\boldmath $ x $}_{\infty}}

\newcommand{\XN}{\mbox{\boldmath $ X $}_N}
\newcommand{\yN}{\mbox{\boldmath $ y $}_N}
\newcommand{\xN}{\mbox{\boldmath $ x $}_N}

\newcommand{\s}{\mbox{\boldmath $ s $}}

\newcommand{\Y}{\mbox{\boldmath $ Y $}}

\def\ipot#1#2{{ \left \{ \matrix{ #1\cr #2\cr } \right. }}

\begin{document}

$ \, $

\bigskip

$ \, $

\centerline{\Large \noindent {\bf On the estimation of the Lorenz curve}}

\centerline{\Large \noindent {\bf under complex sampling designs}}

\bigskip

\centerline{Pier Luigi Conti\footnote{{Pier
Luigi Conti. Dipartimento di Scienze
Statistiche; Sapienza Universit\`{a} di Roma; P.le A.
Moro, 5; 00185 Roma; Italy. E-mail
pierluigi.conti@uniroma1.it}}}

\centerline{Alberto Di Iorio\footnote{{Alberto Di Iorio.
Dipartimento di Scienze
Statistiche; Sapienza Universit\`{a} di Roma; P.le A.
Moro, 5; 00185 Roma; Italy. E-mail
alberto.diiorio@uniroma1.it}}}

\centerline{Alessio Guandalini\footnote{{Alessio Guandalini.
ISTAT, Via Cesare Balbo, 16; 00184 Roma; Italy. E-mail
alessio.guandalini@istat.com}}}

\centerline{Daniela Marella\footnote{{Daniela Marella.
Dipartimento di Scienze della Formazione, Universit\`{a} Roma Tre, via del Castro Pretorio 20; 00185 Roma; Italy. E-mail
daniela.marella@uniroma3.it}}}

\centerline{Paola Vicard\footnote{{Paola Vicard.
Dipartimento di Economia, Universit\`{a} Roma Tre, Via Silvio D'Amico, 77; 00145 Roma; Italy. E-mail
paola.vicard@uniroma3.it}}}

\centerline{Vincenzina Vitale\footnote{{Vincenzina Vitale.
Dipartimento di Economia, Universit\`{a} Roma Tre, Via Silvio D'Amico, 77; 00145 Roma; Italy. E-mail
vincenzina.vitale@uniroma3.it}}}

\bigskip

\centerline{\noindent {\bf Abstract}}

This paper focuses on the estimation of the concentration curve of a finite population, when data are collected according to a complex sampling design with different inclusion probabilities. A (design-based) H\'{a}jek type estimator for the Lorenz curve is proposed, and its asymptotic properties are studied. Then, a resampling scheme able to approximate the asymptotic law of the Lorenz curve estimator is constructed. Applications are given to the construction of $(i)$ a confidence band for the Lorenz curve, $(ii)$ confidence intervals for the Gini concentration ratio, and $(iii)$ a test for Lorenz dominance. The merits of the proposed resampling procedure are evaluated through a simulation study.

\bigskip


\bigskip

\newpage

\section{Introduction}
\label{sec:introduction}

The analysis of income data is fundamental  in both theoretical and applied research.
In particular, a crucial role is played by the Lorenz curve, that essentially consists in  plotting cumulative income shares
against cumulative population shares. The Lorenz curve is a basic tool to construct inequality measures, including the popular Gini
coefficient.
Furthermore,
the comparison of wealth and earnings distributions is
a fundamental  part of income, wealth, and poverty studies, as well as an important
tool for  public economics.

This justifies statistical inference for Lorenz curve and related quantities.
In the literature, since the paper by \cite{gastwirth72}, several papers have been devoted to the subject. Good reviews are in
\cite{giorgi:99}, \cite{giogig17}; cfr. also \cite{csorgo86}, \cite{zheng02}, \cite{bhattacharya07}, \cite{davidson09}.

A basic condition common to many papers is that sampling observations are independent and identically distributed. Unfortunately, this condition is
hardly ever met in practice (cfr. \cite{giorgi:99}, \cite{zheng02}).

The estimation of inequality measures (mainly Gini's index), when data are collected according to a variable probability sampling design from a finite population, is
widely studied in the literature: cfr. \cite{tillelangel13}, \cite{barperri16}, and references therein. However, the same is not true with regard to the estimation
of the {\em whole} Lorenz curve.
In \cite{zheng02}, the asymptotic law of the sample Lorenz curve (computed at a
finite number of points) is obtained under stratified, cluster and multi-stage sampling plans; the sampling fractions within strata are assumed
``small'', so that the finite population correction term is essentially negligible. This is equivalent, of course, to assume that sampling within
strata is simple {\em with} replacement, so that sample data are essentially {\em i.i.d.}.
A step forward is in a couple of papers by Bhattacharya (cfr. \cite{bhattacharya05}, \cite{bhattacharya07}), where asymptotic results for {\em whole}
 Lorenz curve (estimated {\em via} the generalized moment method), under a multi-stage sample design, are obtained. At each stage, units (either primary
 or secondary) are drawn by simple random sampling with replacement. As a consequence, sample data are independent, although not necessarily identically distributed.

Although the above mentioned papers are of the highest importance, the considered sampling designs do not cover several real cases.
For instance, in Italy
reliable income and wealth data come from the
Survey on Household Income and Wealth (SHIW),
conducted by Banca d'Italia (the Italian central bank) every two years.
The sampling design is two-stage, with municipalities and households as primary and secondary sampling
units, respectively. Primary units are stratified by administrative region and population size (less
than 20,000 inhabitants; in between 20,000 and 40,000; 40,000 or more). Within each stratum, primary units
are selected to include all municipalities with a population of 40,000 inhabitants or more; smaller municipalities are selected by using
inclusion probability proportional to size sampling (without replacement). Individual households are then randomly selected,
via simple random sampling without replacement, from
administrative registers. Similar considerations hold for the EU-SILC (Statistics on Income and Living Conditions) survey.

Generally speaking, the sampling design can be dropped whenever the sampling design is ignorable; cfr. \cite{pfeffermann93} and references therein.
In general, a sampling design is ignorable provided that two conditions are met:
\begin{itemize}
\item[{\em Ig 1.}] the sampling design is non-informative, {\em i.e.} the probability of drawing a sample only depends on the values of design variables,
but not on the variable of interest;
\item[{\em Ig 2.}] the values of the design variables are known for all population units.
\end{itemize}

Now, condition {\em Ig 1} is usually satisfied, whilst condition {\em Ig 2} is not, at least for final users of data produced by Official Statistics, since
micro-data are usually released together with sampling weights ({\em i.e.} reciprocals of inclusion probabilities) for sample units only. For instance, this
is exactly what happens in SHIW and EU-SILC.
Ignoring the sample design when it is not ignorable can produce severely biased inference; cfr. the illuminating remarks in \cite{pfeffermann93}.

In the present paper, in view of their importance in applications,
we focus on sampling designs with first inclusion probabilities proportional to a size measure ($\pi${\em ps}
designs). Furthermore, the primary interest is in making inference on the Lorenz curve at a ``superpopulation level''. The results are of
asymptotic nature, with both the population and the sample size increasing. They can be viewed as an
extension of results on the Lorenz curve estimation that are valid in case of {\em i.i.d.} data.

The paper is organized as follows. In Section \ref{sec:problem} the problem is described, and the main assumptions are listed.
In Section \ref{sec:main_basic}, the main asymptotic results are provided. Section \ref{sec:resamp} is devoted to defining the multinomial resampling scheme,
and to establish its properties. Section \ref{sec:confidence_band} is devoted to the construction of a confidence band for the Lorenz curve.
Section \ref{sec:conc_index} focuses on statistical inference for Gini concentration index, and Section \ref{sec:lorenz_dominance} on the construction of a test for Lorenz dominance. Finally, in Section \ref{sec:simulation} a
simulation study is performed.

\section{The problem}
\label{sec:problem}

\subsection{Superpopulation model}
\label{sec:superpopulation}

Let ${\mathcal U}_N$ be a finite population of size $N$. If $Y$ denotes a non-negative character of interest, let $y_i$ be the value of character $Y$ for
unit $i$ ($=1, \, \dots , \, N$). Each $y_i$ value is assumed to be a realization of a random variable (r.v.) $Y_i$; the $N$-variate r.v. $\Y_N = (Y_1 \, \cdots \, Y_N )$
is the {\em superpopulation}.

In the sequel, the r.v.s $Y_i$ are assumed to be independent and identically distributed ({\em i.i.d.}), and their distribution function (d.f.) is denoted by
\begin{eqnarray}
F(y) = Pr (Y_i \leq y) . \label{eq:superpop_df}
\end{eqnarray}
The superpopulation quantile of order $p$, with $0 < p < 1$, is defined as
\begin{eqnarray}
Q(p) = \inf \{ y: \; F(y) \geq p \} . \label{eq:superp_quantile}
\end{eqnarray}
The superpopulation generalized Lorenz curve $G( \cdot )$ is obtained by integrating the quantile function $Q( \cdot )$. In symbols
\begin{eqnarray}
G(p) = \int_{0}^{p} Q(u) \, du , \;\; 0 \leq p \leq 1. \label{eq:superp_gen_lorenz}
\end{eqnarray}
The curve $G( \cdot )$ is continuous, increasing, convex, with $G(0)=0$ and
\begin{eqnarray}
G(1)=  \mathbb{E} [ Y_i ] = \mu_Y ,
\label{eq:superpopol_g1}
\end{eqnarray}
\noindent the superpopulation mean.

The superpopulation Lorenz curve is the normalized version of $( \ref{eq:superp_gen_lorenz} )$, namely
\begin{eqnarray}
L(p)= \frac{G(p)}{G(1)} = \frac{G(p)}{\mu_{Y}} , \;\; 0 \leq p \leq 1.
\label{eq:superp_lorenz}
\end{eqnarray}
Of course, $L( \cdot )$ is convex, continuous, increasing, with $L(0)=0$, $L(1)=1$.

In the sequel, attention is devoted to the estimation of $L(p)$. Due to the effect of the sampling design, even if the r.v.s
$Y_i$s are {\em i.i.d.} at a (super)population level, they are not {\em i.i.d.} at a sample level (except very special cases); cfr. \cite{pfeffermann93}.

Alongside $( \ref{eq:superpop_df} )$-$( \ref{eq:superp_lorenz} )$, one may define the corresponding {\em finite population} counterparts.
The finite population distribution function (p.d.f., for short) is defined as
\begin{eqnarray}
F_N (y) = \frac{1}{N} \sum_{i=1}^{N} I_{(- \infty , \, y]} (y_i) , \;\; y \in \mathbb{R}
\label{eq:popul_distr_fc}
\end{eqnarray}
\noindent where
\begin{eqnarray}
I_A (y) = \ipot{1 \; {\mathrm {if}} \; y \in A}{0 \; {\mathrm {if}} \; y \notin  A}  .
\nonumber
\end{eqnarray}
\noindent Clearly, $F(y)$ is the expectation of $F_N (y)$ w.r.t. the superpopulation probability distribution: $F(y) = \mathbb{E} [ F_N (y) ]$.

The finite population quantile of order $p$ is
\begin{eqnarray}
Q_N (p) = \inf \{ y: \; F_N (y) \geq p \} , \;\; 0 < p < 1
\label{eq:p-quantile}
\end{eqnarray}

Next, the finite population
{\em generalized Lorenz curve} is defined as
\begin{eqnarray}
G_N (p) = \int_{0}^{p} Q_N (u) \, du , \;\; 0 \leq p \leq 1 .
\label{eq:generalized_lorenz}
\end{eqnarray}
Clearly, $G_N ( \cdot )$ is continuous, increasing, and convex, with $G_N (0) = 0$ and
\begin{eqnarray}
G_N (1) = \overline{Y}_N = \frac{1}{N} \sum_{i=1}^{N} y_i
\label{eq:max_G}
\end{eqnarray}
\noindent  $\overline{Y}_N$ being the {\em finite population mean}.

The {\em Lorenz curve}, in its turn, is the normalized version of $( \ref{eq:generalized_lorenz} )$, namely
\begin{eqnarray}
L_N (p) = \frac{G_N (p)}{G_N (1)} = \frac{G_N (p)}{\overline{Y}_N}  , \;\; 0 \leq p \leq 1 .  \label{eq:lorenz-curve}
\end{eqnarray}
\noindent
Of course, $L_N (p)$ is increasing, continuous, convex, with $L_N (0) = 0$,  $L_N (1) =1$.

As already said, our main interest is in estimating the Lorenz curve at a superpopulation level. However, the estimation of
$( \ref{eq:popul_distr_fc} )$-$( \ref{eq:lorenz-curve} )$ will play an important, although indirect, role.

\subsection{Sampling design and superpopulation model: basic aspects}
\label{sec:superpopul_model}

In general, a sample is a subset of the population ${\mathcal U}_N$.
For each unit $i \in {\mathcal U}_N$, define a  Bernoulli random variable (r.v.) $D_i$ such that $i$ is  (is not) in the sample whenever  $D_i =1$ ($D_i =0$);   denote further by $\DN$  the $N$-dimensional vector of components $D_1$, $\dots$, $D_N$.
A (unordered, without replacement) sampling design $P$ is the probability distribution of the random vector $\DN$. The expectations  $\pi_i = E_P [ D_i ]$ and
$\pi_{ij} = E_P [ D_i  \, D_j ]$ are the first and second order inclusion probabilities, respectively. The suffix $P$ denotes the sampling design used to select population units. The sample  size is $n_s = D_1 + \cdots + D_N$. In the present paper we focus on fixed size sampling designs, such that $n_s \equiv  n$.

In practice, the sampling design is constructed on the basis of the value of the {\em design variables}, {\em i.e.} auxiliary variables known for all
population units (cfr. \cite{pfeffermann93}). In particular, the first order inclusion probabilities are frequently chosen to be proportional to an auxiliary variable $X$, depending itself on the design variables: $\pi_i \propto x_i$, $i=1, \, \dots , \, N$. A special case is the stratified sampling design, where $x_i$ is proportional to the weight of the stratum containing unit $i$.
The rationale of the choice $\pi_i \propto x_i$
is simple: if the values of the variable of interest are positively correlated with (or, even better, approximately proportional to) the values of $X$, then the Horvitz-Thompson estimator of the population mean will be highly efficient.

For each unit $i$, let $p_i$ be a positive number, with  $p_1 + \cdots + p_N =n$. The {\em Poisson sampling design} ($Po$, for short) with parameters $p_1$, $\dots $, $p_N$ is characterized by the independence of the r.v.s $D_i$s, with $Pr_{Po} ( D_i =1 ) = p_i$. In symbols
\begin{eqnarray}
Pr_{Po} ( \DN ) = \prod_{i=1}^{N} p_i^{D_{i}} (1- p_i )^{1- D_{i}} .
\nonumber
\end{eqnarray}
The {\em rejective sampling} ($P_R$), or {\em normalized conditional Poisson sampling}
(\cite{hajek64}, \cite{tille06}) corresponds to the probability distribution of the random vector $\DN$, under Poisson design, conditionally on $n_s =n$.

The {\em Hellinger distance} between a sampling design $P$ and the rejective design is defined as
\begin{eqnarray}
d_H (P, \, P_R ) = \sum_{D_1 , \, \dots , \, D_N} \left (
\sqrt{Pr_P ( \DN )}  - \sqrt{Pr_R ( \DN )} \right )^2 .
\label{eq:hellinger}
\end{eqnarray}

For each $N$, $( y_i , \, t_{i1} , \, \dots , \, t_{iL} )$, $i=1, \, \dots , \, N$ are
 realizations of a superpopulation $\{ ( Y_i , \, T_{i1} , \, \dots , \, T_{iL} ) , \; i=1 , \, \dots , \, N\}$ composed by
{\em i.i.d.} $(L+1)$-dimensional r.v.s.
In the sequel, the symbol $\mathbb{P}$ will denote the (superpopulation) probability distribution of r.v.s $( Y_i , \, T_{i1} , \, \dots , \, T_{iL} )$s, and
 $\mathbb{E}$, $\mathbb{V}$ are the corresponding operators of mean and variance, respectively.

It is important to observe that $Y_i$s are assumed marginally {\em i.i.d.}. Conditionally on $ T_{i1} , \, \dots , \, T_{iL} $, $Y_i$s are still independent, but not necessarily identically distributed. This covers, among others, the important case of stratified populations.

\subsection{Assumptions}
\label{sec:assumptions}

Our assumptions on both the superpopulation model and sampling design are listed below.
\begin{itemize}
\item[A1.] $( {\mathcal U}_N ; \; N \geq 1)$ is a sequence of finite populations of increasing size $N$.
\item[A2.] For each $N$, $( y_i , \, t_{i1} , \, \dots , \, t_{iL} )$, $i=1, \, \dots , \, N$ are
 realizations of a superpopulation $\{ ( Y_i , \, T_{i1} , \, \dots , \, T_{iL} ) , \; i=1 , \, \dots , \, N\}$ composed by
{\em i.i.d.} $(L+1)$-dimensional r.v.s, with $Y_i \geq 0$ almost surely.
\item[A3.] The d.f. $F(y)$ $( \ref{eq:superpop_df} )$ is continuously differentiable, with density function $f(y) = d F(y) / dy$ strictly positive of every
interval $[a, \, b]$ with $a>0$, $b < \infty$. Furthermore, $E[ Y_i^2] < \infty$ and:
\begin{eqnarray}
\lim_{y \downarrow 0} \frac{F(y)^{\gamma}}{f(y)} = 0 , \;\;  \lim_{y \uparrow \infty} \frac{(1-F(y))^{\gamma}}{f(y)} = 0
\label{eq:condizione_code}
\end{eqnarray}
for some $0 < \gamma < 1$.
\item[A4.] For each population ${\mathcal U}_N$, sample units are selected according to a fixed size sample design  with positive first order inclusion probabilities
$\pi_1 $, $\dots $, $\pi_N$, and sample size $n = \pi_1 + \cdots + \pi_N$. The first order inclusion probabilities are taken proportional to $x_i =
h ( t_{i1} , \, \dots , \, t_{iL} )$, $i=1, \, \dots , \, N$, $h( \cdot )$ being an arbitrary (positive) function. To avoid complications in the notation, we will assume that
$\pi_i = n x_i / \sum_{i=1}^{N} x_i$ for each unit $i$.
It is also assumed that
\begin{eqnarray}
\lim_{N, n \rightarrow \infty} \mathbb{E} [ \pi_i ( 1- \pi_i ) ] = d >0 .  \label{eq:quant_d}
\end{eqnarray}
Furthermore, the notation $\xN = (x_1 , \dots , \, x_N )$ is used.
\item[A5.]
The sample size $n$ increases as the population size $N$ does, with
\begin{eqnarray}
\lim_{N \rightarrow \infty} \frac{n}{N} = f, \;\; 0 < f < 1 .
\nonumber
\end{eqnarray}
\item[A6.] For each population $( {\mathcal U}_N ; \; N \geq 1)$, let $P_R$ be
the rejective sampling design with inclusion probabilities $\pi_1$, $\dots$, $\pi_N$, and let $P$ be the actual sampling design
(with the same inclusion probabilities). Then
\begin{eqnarray}
d_H (P, \, P_R ) \rightarrow 0 \;\; {\mathrm as} \; N \rightarrow \infty , \;\; a.s.-{\mathbb{P}}. \nonumber
\end{eqnarray}
\item[A7.] $\mathbb{E} [ X_1^2 ] < \infty$, so that the quantity in $(  \ref{eq:quant_d} )$ is equal to:
\begin{eqnarray}
 d = f \left ( 1- \frac{\mathbb{E} [ X_1^2 ]}{\mathbb{E} [ X_1 ]^2} \right ) +
f (1-f) \frac{\mathbb{E} [ X_1^2 ]}{\mathbb{E} [ X_1 ]^2} > 0 .
\label{eq:val_d}
\end{eqnarray}
\end{itemize}

\section{Basic asymptotic results}
\label{sec:main_basic}

Due to the effect of the sampling design, the empirical distribution function (e.d.f.)
\begin{eqnarray}
F_n (y) = \frac{1}{n} \sum_{i=1}^{N} I_{(- \infty , \, y]} (Y_i)
\label{eq:emp_df}
\end{eqnarray}
\nonumber is {\em inconsistent}.
In fact, in view of the law of large numbers,
\begin{eqnarray}
F_n (y) & \stackrel{p}{\rightarrow} & \mathbb{E} \left [ \frac{X_{i}}{\mathbb{E}[X_{i}]}
I_{(- \infty , \, y]} (Y_i ) \right ] \nonumber \\
\, & = & \frac{1}{\mathbb{E}[X_{i}]} \mathbb{E} \left [ X_i I_{(- \infty , \, y]} (Y_i ) \right ] \nonumber \\
\, & \neq & F(y) \nonumber
\end{eqnarray}
\noindent unless $X_i$ and $Y_i$ are independent.
As a consequence, the {\em empirical Lorenz curve} studied, for instance, in \cite{csorgo86}, is {\em inconsistent}, too.

The first, basic step consists in constructing a consistent estimator of $F(y)$, and then in studying the
asymptotic distribution of the corresponding estimate of $L(p)$.
The d.f. $F$ is estimated by using the (design-based) H\'{a}jek estimator
\begin{eqnarray}
\widehat{F}_{H} (y) & = & \frac{\sum_{i=1}^{N} \frac{1}{\pi_{i}} D_i I_{(- \infty , \, y]} (y_i )}{\sum_{i=1}^{N} \frac{1}{\pi_{i}} D_i}  ,
\label{eq:dfhajek}
\end{eqnarray}
\noindent that generates the corresponding ``empirical process''
\begin{eqnarray}
\mathcal{W}_{N}^{H} ( \cdot ) = \{ \sqrt{n} ( \widehat{F}_H (y) - F (y) )  , \: y \in \mathbb{R} \} ; \;\; N \geq 1 . \label{eq:empirical_hajek}
\end{eqnarray}

The weak convergence properties of  $( \ref{eq:empirical_hajek} )$ are studied in \cite{contiorio18}, \cite{bastard17}.

Denote by
\begin{eqnarray}
S ( y, \, x) = \mathbb{P} ( Y_i \leq y , \, X_i \leq x )
\label{eq:joint_sdf}
\end{eqnarray}
\noindent the joint superpopulation d.f. of $( Y_i , \, X_i )$, and by
\begin{eqnarray}
F (y) = \mathbb{P} ( Y_i \leq y ) = S( y , \, + \infty ) , \;\;
M (x) = \mathbb{P} ( X_i \leq x ) = S( + \infty , \, x ) ,
\label{eq:marginal_sdf}
\end{eqnarray}
\noindent  the marginal superpopulation d.f.s of $Y_i$ and $X_i$, respectively. Furthermore,
from now on the notation
\begin{eqnarray}
K_{\alpha} (y) = \mathbb{E} \left [ \left . X_1^{\alpha} \, \right \vert Y_1 \leq y \right ] , \;\;
y \in {\mathbb{R}}, \; \alpha = 0, \, \pm 1, \, \pm 2 \label{eq:def_K}
\end{eqnarray}
\noindent will be used. Note that $K_{\alpha} ( + \infty ) = {\mathbb{E}} [ X_1^{\alpha} ]$.

\begin{proposition}
\label{asympt_hajek}
Assume that conditions A1-A7 are satisfied, and define
\begin{eqnarray}
\mathcal{W}_{1N}^{H} ( \cdot ) = \{ \sqrt{n} ( \widehat{F}_H (y) - F_N (y) )  , \: y \in \mathbb{R} \} ; \;\; N \geq 1 . \label{eq:empirical_hajek1} \\
\mathcal{W}_{2N}^{H} ( \cdot ) = \{ \sqrt{n} ( F_N (y) - F (y) )  , \: y \in \mathbb{R} \} ; \;\; N \geq 1  \label{eq:empirical_hajek2}
\end{eqnarray}
\noindent so that $\mathcal{W}_{N}^{H} ( \cdot ) = \mathcal{W}_{1N}^{H} ( \cdot ) + \mathcal{W}_{2 N}^{H} ( \cdot )$.
Define further
\begin{eqnarray}
C_1 (y, \, t) & = &
 f \left \{ \frac{\mathbb{E} [ X_1 ]}{f} K_{-1} (y \wedge t) -1  \right \} F( y \wedge t) \nonumber \\
\, & \, &
- \frac{f^3}{d} \left ( 1 - \frac{K_1 (y)}{\mathbb{E} [ X_1 ]} \right )
\left ( 1 - \frac{K_1 (t)}{\mathbb{E} [ X_1 ]} \right ) F(y) F(t) \nonumber \\
\, & \, & - f
\left \{ \frac{\mathbb{E} [ X_1 ]}{f} \left ( K_{-1} (y) + K_{-1} (t) - \mathbb{E} \left [
X_1^{-1} \right ] - 1
\right )
\right \} F(y) F(t) ,
\label{eq:cov_ker_hajek}
\end{eqnarray}
\noindent with $d$ given by $( \ref{eq:val_d} )$, and
\begin{eqnarray}
C_2 (y, \, t) = F (y \wedge t) - F(y) F(t) .
\label{eq:nucleo_brownian}
\end{eqnarray}
Then, the following statements hold.
\begin{itemize}
\item The sequence $( \mathcal{W}_{1N}^{H} ( \cdot ) ; \; N \geq 1) $ converges weakly, in $D[ - \infty , \: + \infty ]$ equipped with the Skorokhod topology,
to a Gaussian process $\mathcal{W}^H_1 ( \cdot )$ with  zero mean function and covariance kernel $C_1 (y, \, t)$ $( \ref{eq:cov_ker_hajek} )$.
\item The sequence $( \mathcal{W}_{2N}^{H} ( \cdot ) ; \; N \geq 1) $ converges weakly, in $D[ - \infty , \: + \infty ]$ equipped with the Skorokhod topology,
to a Gaussian process $\mathcal{W}^H_2 ( \cdot )$ with  zero mean function and covariance kernel $f C_2 (y, \, t)$ $( \ref{eq:nucleo_brownian} )$.
\item The two sequences $( \mathcal{W}_{1N}^{H} ( \cdot ) ; \; N \geq 1) $, $( \mathcal{W}_{2N}^{H} ( \cdot ) ; \; N \geq 1) $ are asymptotically independent,
so that
the sequence $ ( \mathcal{W}_{N}^{H} ( \cdot) ; \; N \geq 1) $, converges weakly, in $D[ - \infty , \: + \infty ]$ equipped with the Skorokhod topology,
to a Gaussian process
$\mathcal{W}^H ( \cdot ) = (\mathcal{W}^H (y) ; \; y \in \mathbb{R} )$ with zero mean function and covariance kernel
\begin{eqnarray}
C(y, \, t) = C_1 (y, \, t) + f C_2 (y, \, t)
\label{eq:nucleo_cov}
\end{eqnarray}
\end{itemize}
\end{proposition}
\begin{proof} See \cite{contiorio18} or \cite{bastard17}.
\end{proof}

In particular, if $F(y) $ is continuous and the sampling design is simple random sampling without replacement of size $n$,
the H\'{a}jek estimator $( \ref{eq:dfhajek} )$
reduces to the empirical d.f. $F_n (y) = n^{-1} \sum_{i} D_i I_{( - \infty , \, y]} ( y_i )$. Furthermore,  in this case
$C_1 (y, \, t) = (1-f) (F( y \wedge t) - F (y) F(t))$. Hence, if $F(y) $ is continuous and the sampling design is simple random sampling without replacement,
$C(y, \, t) = F(y \wedge t) - F(y) F(t)$, and the limiting process $W( \cdot )$ can be represented as $\mathcal{W}^H (y) =  \mathcal{B} (F(y))$,
$\mathcal{B} ( \cdot ) = \{ \mathcal{B} (p) ; \; 0 \leq p \leq 1 \}$ being a Brownian bridge.

The term $( \ref{eq:cov_ker_hajek} )$ can be equivalently re-written as:
\begin{eqnarray}
C_1 (y, \, t) & = &
 \left \{\mathbb{E} [ X_1 ] T_{-1} (y \wedge t) - f  F( y \wedge t) \right \}  \nonumber \\
\, & \, &
- \frac{f^3}{d} \left ( \frac{T_1 (y)}{\mathbb{E} [ X_1 ]} - F(y) \right )
\left (  \frac{T_1 (t)}{\mathbb{E} [ X_1 ]} - F(t)  \right ) \nonumber \\
\, & \, & -
\left \{ \mathbb{E} [ X_1 ] \left ( T_{-1} (y) F(t) + T_{-1} (t) F(y) - \left ( \mathbb{E} \left [
X_1^{-1} \right ] + 1 \right ) F(y) F(t)
\right )
\right \}  ,
\label{eq:cov_ker_hajek2}
\end{eqnarray}
\noindent where
\begin{eqnarray}
T_{\alpha} (y) & = &  \mathbb{E} \left [ X_1^{\alpha} I_{( Y_{1} \leq y)}   \right ] \nonumber \\
\, & = & \int_{0}^{y}  \mathbb{E} \left [ \left .  X_1^{\alpha} \, \right \vert Y_1 =u  \right ] \, d F(u)
, \;\; \alpha = 0, \pm 1 .
\label{eq:def_talfa}
\end{eqnarray}

The map $ y \mapsto T_{\alpha} (y)$ is monotone non-decreasing, right continuous, with $T_{\alpha} (y) \downarrow 0$ as $y \downarrow 0$ and
$T_{\alpha} (y) \uparrow \mathbb{E} [ X_1^{\alpha} ]$ as $y \uparrow + \infty$. Hence, $T_{\alpha} (y)$ induces a finite measure on the real line
(equipped with the Borel $\sigma$-field). In view of $( \ref{eq:def_talfa} )$, such a measure is absolutely continuous w.r.t. the probability
measure induced by $F(y)$, and $\mathbb{E} \left [ \left .  X_1^{\alpha} \, \right \vert Y_1 =y  \right ]$ is the corresponding
Radon-Nikodym derivative. In symbols:
\begin{eqnarray}
T_{\alpha} ( \cdot ) \ll F( \cdot ) \; {\mathrm{with}} \; \frac{d T_{\alpha} (y)}{d F(y)} = \mathbb{E} \left [ \left .
 X_1^{\alpha} \, \right \vert Y_1 = y  \right ] . \label{eq:radon}
\end{eqnarray}

Next assumption C1 ensures that the trajectories of the limiting process $\mathcal{W}^H ( \cdot )$ behave regularly, {\em i.e.} that they are (uniformly) continuous and bounded over the real line.

\begin{itemize}
\item[C1.] The conditional expectation  $\mathbb{E} \left [ \left .
 X_1^{\alpha} \, \right \vert Y_1 = y  \right ]$ is bounded w.r.t. $y$:
\begin{eqnarray}
\left \vert \mathbb{E} \left [ \left .
 X_1^{\alpha} \, \right \vert Y_1 = y  \right ] \right \vert \leq M_{\alpha} \;\; \forall \, y ; \; \alpha = \pm1 . \label{eq:radon-bounded}
\end{eqnarray}
\end{itemize}

\begin{proposition}
\label{continuous_traject}
Define the Gaussian process $ \mathcal{B}^H (p) = \mathcal{W}^H ( Q(p) )$, with $0 \leq p \leq 1$. If $F(y)$ is continuous, then
the process $\mathcal{B}^H$ possesses with probability 1 trajectories that are continuous (and bounded) in $[0, \, 1]$.
\end{proposition}
\begin{proof} See Appendix.
\end{proof}

The process $ \mathcal{B}^H $ shares several properties with the Brownian bridge: it is a Gaussian process with a.s. continuous trajectories, and with
$B^H (0) = B^H (1) =0$ with probability 1.

Next assumption D1 is essentially the same as in  \cite{bhattacharya07}.

\begin{itemize}
\item[D1.] The density $f(y) = F^{\prime} (y)$ exists, is positive, and satisfies the relationships
$(1- F(y))^{1+ \gamma} / f(y) \rightarrow 0$ as $y \rightarrow \infty$ and
$F(y)^\gamma / f(y) \rightarrow 0$ as $y \rightarrow 0$, for some $0 < \gamma < 1$. Furthermore, $\mathbb{E} [ Y_i^{2+ \epsilon} ] < \infty$ for some positive $\epsilon$.
\end{itemize}

Proposition $\ref{continuous_traject}$ and assumption D1 allow one to use the same reasoning as in \cite{bhattacharya07}, and to show that the map $F ( \cdot ) \mapsto G( \cdot )$ is Hadamard differentiable at $F$ tangentially to the space of the trajectories of the limiting process $\mathcal{W}^H ( \cdot )$.
The Hadamard derivative (computed at ``point'' $h$) is equal to:
\begin{eqnarray}
\int_{0}^{p} \frac{h(Q(u))}{f(Q(u))} \, du .
\label{eq:Hadam_genlorenz}
\end{eqnarray}

As an application  of Theorem 20.8 in \cite{vandervaart98}, we are now in a position to obtain the following result.

\begin{proposition}
\label{asympt_hadam}
Under assumptions A1-A6, C1, D1, the following weak convergence results hold:
\begin{eqnarray}
\mathcal{G}^H_N ( \cdot ) & = & \sqrt{n} ( \widehat{G}_H ( \cdot ) - G ( \cdot ))  \stackrel{w}{\rightarrow}
\mathcal{G}^H ( \cdot ) \;\; as \; N \rightarrow \infty , \;\; a.s.-\mathbb{P} ; \label{eq:conv_1} \\
\mathcal{L}^H_N ( \cdot ) & = & \sqrt{n} ( \widehat{L}_H ( \cdot ) - L ( \cdot ))  \stackrel{w}{\rightarrow}
\mathcal{L}^H ( \cdot ) \;\; as \; N \rightarrow \infty  , \;\; a.s.-\mathbb{P} \label{eq:conv_2} .
\end{eqnarray}
\noindent where $\mathcal{G}^H ( \cdot )$, $\mathcal{L}^H ( \cdot )$ are Gaussian processes that can be
represented as:
\begin{eqnarray}
\mathcal{G}^H ( p ) & = & \int_{0}^{p} \frac{\mathcal{W}^H ( Q(u))}{f(Q(u))} \, du , \;\; 0 \leq p \leq 1; \label{eq:proc_g} \\
\mathcal{L}^H ( p ) & = & G(1)^{-1} ( \mathcal{G}^H ( p ) + L(p) \mathcal{G}^H (1 ) )  , \;\; 0 \leq p \leq 1. \label{eq:proc_l}
\end{eqnarray}
\noindent respectively.
\end{proposition}

The process $\mathcal{L}^H ( \cdot )$ is, in a sense, the finite population counterpart of the concentration process studied, in case of {\em i.i.d.} data,
in \cite{goldie77}, \cite{csorgo86}.

The limiting Gaussian processes $W^H ( \cdot )$, $\mathcal{G}^H ( \cdot )$, $\mathcal{L}^H  ( \cdot )$ are quite non-standard. Their
covariance kernels are complicate, and, ever worse, they depend on the unknown quantities $F$, $Q$, $S$, $G$, $L$.
For this reason, in the subsequent section a resampling procedure to approximate the probability law of the processes $\mathcal{G}^H_N ( \cdot)$,
$\mathcal{L}^H_N ( \cdot)$ is developed.

Before ending this section, we note  in passing that from Proposition $\ref{asympt_hadam}$ it is also possible to obtain, virtually with no additional
effort, the limiting distribution of the Gini concentration index
\begin{eqnarray}
R = 1-2 \int_{0}^{1} L(p) \, dp . \label{eq:gini_index}
\end{eqnarray}
\noindent
Consider in fact its estimator
\begin{eqnarray}
\widehat{R}_H = 1-2 \int_{0}^{1} \widehat{L}_H (p) \, dp .
\label{eq:gini_estimator}
\end{eqnarray}
From $( \ref{eq:gini_index} )$ it appears that $R$ is a linear functional of $L( \cdot )$, so that it is Hadamard differentiable. As a consequence of the chain rule for Hadamard derivatives (see, {\em e.g.}, \cite{vandervaart98}), the map $ F \mapsto R$ is Hadamard differentiable, too, with Hadamard derivative (computed at `point'' $h$):
\begin{eqnarray}
-2 \int_{0}^{1} \left \{ \int_{0}^{p} \frac{h(Q(u))}{f(Q(u))} du \right \} dp .
\nonumber
\end{eqnarray}
Taking into account that linear functionals of Gaussian processes possess normal distribution, from Proposition $\ref{asympt_hadam} $ the following result follows.

\begin{proposition}
\label{asympt_R} Under the assumptions of Proposition $\ref{asympt_hadam}$ the limiting distribution, as $n$, $N$ tend to infinity, of
$ \sqrt{n} ( \widehat{R}_H - R )$ can be represented as
\begin{eqnarray}
- 2 \int_{0}^{1} \mathcal{L}^H (p) \, dp . \label{eq:limit_R}
\end{eqnarray}
The probability law of $( \ref{eq:limit_R} )$ turns out to be normal with zero expectation and variance
\begin{eqnarray}
\sigma^2_{R,H} = 4 \int_{0}^{1} \int_{0}^{1} E[ \mathcal{L}^H (u) \mathcal{L}^H (v) ] \, du \, dv .
\label{eq:variance_R}
\end{eqnarray}
\end{proposition}

\section{The resampling procedure}
\label{sec:resamp}

In the present section we develop a resampling procedure to approximate the distribution of $\widehat{G}_H ( \cdot )$ and
$\widehat{L}_H ( \cdot )$.
The basic requirement is to recover the limiting laws obtained in Propositions \ref{asympt_hajek}, \ref{asympt_hadam}.
In other words, we aim at constructing a resampling procedure that is asymptotically exact.
This is in fact the main justification of classical Efron's bootstrap for {\em i.i.d.} data; see, {\em e.g.}, \cite{bickelfried81}.
Unfortunately, in the present case
classical bootstrap does not work, because of the dependence among units due to the sampling design. This fact is well-known in the literature on sampling
finite populations: cfr. \cite{antal:11}, \cite{chauvet07}, \cite{contimar14} and references therein.

In sampling finite populations several different resampling techniques exist, but none of them possesses a true asymptotic justification. The only
exception is the method developed in \cite{contimar14}, that unfortunately is not suitable in the case under examination, because it assumes the absence of
relationships between the sampling weights $\pi_i^{-1}$ and the values $y_i$s of the variable of interest.

The resampling procedure we consider here has been proposed in \cite{contiorio18}, and exploited in \cite{marvicard18}.
It  is composed by two phases. In the first one,
on the basis of the sampling data a pseudo-population, consisting in a prediction of the ``true'' population, is constructed.
The prediction process is based on the sampling design, and does not essentially involve the superpopulation model. In the second phase, a sample of size $n$ (the same as the ``original'' one) is drawn from the pseudo-population, according to a $\pi ps$ sample design $P^*$ (the {\em resampling design}) with inclusion probabilities appropriately chosen and satisfying the entropy condition A5.

From now on, the following terminology will be used. The {\em sampling design} $P$ is the sampling procedure drawing $n$ units from the ``original''
population $\mathcal{U}_N$. The {\em resampling design} $P^*$ is the sampling procedure drawing $n$ units from the pseudo-population.

\subsection{Pseudo-population}
\label{sec:resampling_1b}

A {\em design-based population predictor} of $\yN$ is
\begin{eqnarray}
\left \{ ( N^{*}_i D_i , \, y_i , \, x_i  ) ; \; i  =1, \, \dots , \, N \right \}
\label{eq:predict_des}
\end{eqnarray}
\noindent where $N^{*}_i$s are integer-valued r.v.s, with (joint) probability distribution $P_{pred}$.
In practice, $( \ref{eq:predict_des} )$ means that
$N^{*}_i D_i$ population units are predicted to have $y$-value equal to $y_i$ and
$x$-value equal to $x_i$, for each sample unit $i$.
In the sequel, the symbols $y^{*}_k$, $x^{*}_k$ will be used to denote the $y$-value and $x$-value of unit $k$
of the pseudo-population, respectively. Of course $N^{*}_i$ units of the pseudo-population satisfy the relationships $y^{*}_k = y_i$,
$x^{*}_k = x_i$, $i \in \s$.

Although several pseudo-populations could be constructed, according to \cite{contiorio18} there is essentially only one pseudo-population
that asymptotically works in a superpopulation perspective: the {\em multinomial pseudo-population}. In a non-asymptotic setting, it goes back to
 \cite{pfeffersver04}.

Consider $N$ independent trials, where trial $k$ ($=1, \, \dots , \, N$) consists in choosing a unit from the original sample $\s$;  unit $i \in \s$ is selected
 with probability $\pi_i^{-1} / \sum_{j \in \s} \pi_j^{-1} = x_i^{-1} / \sum_{j \in \s} x_j^{-1}$.
If at trial $k$ the unit $i \in \s$ is selected, define $y^*_k = y_i$ and $x^*_k = x_i$, $k=1, \, \dots , \, N$. Next,
define a pseudo-population of $N$ units, such that unit $k$ possesses $y$-value $y^*_k$ and $x$-value $x^*_k$, $k=1 \, \dots , \, N$.
Finally, let $N^*_i$, $i \in \s$, be the number of the pseudo-population units equal to unit $i$ of the sample $\s$. Of course, the pseudo-population has size $N$.

\subsection{Resampling scheme}
\label{sec:resampling_1c}

The resampling procedure we consider is described below.

\begin{itemize}
\item[Ph 1.] Generate a pseudo-population $( \ref{eq:predict_des} )$ of $N$ units. Denote by $y^{*}_k$, $x^{*}_k$ the $y$-value and $x$-value of unit $k$ of the pseudo-population, respectively.
\item[Ph 2.]
Draw a sample $\s^*$ of size $n$ from the pseudo-population defined in phase 1, on the basis of a resampling design $P^*$ with first order inclusion probabilities $\pi_k^* = n x^*_k / \sum_{h=1}^{N} x^*_h$ and satisfying assumption A5.
\end{itemize}

Consider now the resampling design, and let $D^{*}_{k} =1$ if the unit $k$ of the pseudo-population is drawn, and $D^{*}_k =0$ otherwise. The H\'{a}jek estimator of
$F^*_N (y)$ is equal to
\begin{eqnarray}
\widehat{F}^*_H (y) = \frac{\sum_{k=1}^{N} \frac{D^{*}_{k}}{\pi^{*}_{k}} I_{( y^{*}_{k} \leq y)}}{\sum_{k=1}^{N} \frac{D^{*}_{k}}{\pi^{*}_{k}}} .
\label{eq:hajek_resampled}
\end{eqnarray}

Next, define the resampled version of the process
$( \ref{eq:empirical_hajek} )$, namely
\begin{eqnarray}
\mathcal{W}_{N}^{H*} (y) = \sqrt{n} ( \widehat{F}^{*}_H (y) - \widehat{F}_H (y) )  , \:\: y \in \mathbb{R} ; \;\; N \geq 1 . \label{eq:empirical_resampled_hajek}
\end{eqnarray}

The main property of the above resampling scheme is its asymptotic correctness.

\begin{proposition}
\label{asympt_ricamp_hajek}
Assume the sampling design $P$ and the resampling design $P^*$ both satisfy  assumptions A1-A6, and that
P1-P3 are fulfilled. Conditionally on $\yN$, $\xN$, $\DN$, $( D_1 N^*_1 , \, \dots , \, D_N N^*_N)$,
the following statements hold.
\begin{itemize}
\item[$PR1$.] The
 sequence $ ( \mathcal{W}_{N}^{H*} ( \cdot) ; \; N \geq 1) $, converges weakly, in $D[ - \infty , \: + \infty ]$ equipped with the Skorokhod topology,
to a Gaussian process
$\mathcal{W}^H ( \cdot ) $ with zero mean function and covariance kernel $( \ref{eq:nucleo_cov} )$.
\item[$PR2$.] If $\phi ( \cdot )$ is  Hadamard differentiable at $F$, then  $( \sqrt{n} ( \phi ( \widehat{F}^{*}_H (y) ) - \phi (
\widehat{F}_H )) ; \; N \geq 1) $
converges weakly to $\phi^{\prime}_F ( \mathcal{W}^H )$, as $N$ increases.
\end{itemize}
In both $PR1$, $RP2$ weak convergence takes place for a set of $y_i$s, $x_i$s having $\mathbb{P}$-probability 1,
and for a set of $\DN$s and $( N^*_1 , \, \dots , \, N^*_N)$ of probability tending to $1$.
\end{proposition}
\begin{proof}
See \cite{contiorio18}.
\end{proof}

In particular, denote by $\widehat{G}^{*}_{H}$, $\widehat{L}^{*}_{H}$ the estimators of the  generalized Lorenz curve and Lorenz curve, respectively,
based on the sample $\s^*$ drawn from the pseudo-population.
Consider further the resampled processes
\begin{eqnarray}
\mathcal{G}_{N}^{*H} ( \cdot ) = \sqrt{n} ( \widehat{G}^{*}_{H} (p) - \widehat{G}_{H} (p) ) , \;\;\;
\mathcal{L}_{N}^{*H} ( \cdot ) = \sqrt{n} ( \widehat{L}^{*}_{H} (p) - \widehat{L}_{H} (p) ) ; \;\; 0 \leq p \leq 1.
\label{eq:resampled_proc}
\end{eqnarray}
From Propositions $\ref{asympt_hadam}$ and $\ref{asympt_ricamp_hajek}$ it is immediate to obtain the following corollary.

\begin{corollary}
\label{asymp_resampled_lorenz}
Under the assumptions of Propositions $\ref{asympt_hadam}$, $\ref{asympt_ricamp_hajek}$, the resampled processes $( \ref{eq:resampled_proc} )$
asymptotically behave as the processes $\mathcal{G}_N^H ( \cdot )$, $\mathcal{L}_N^H ( \cdot )$ in $( \ref{eq:conv_1} )$,
$( \ref{eq:conv_2} )$. In symbols:
\begin{eqnarray}
\mathcal{G}_{N}^{*H} ( \cdot ) \stackrel{w}{\rightarrow} \mathcal{G}^H ( \cdot ) , \;\;
\mathcal{L}_{N}^{*H} ( \cdot ) \stackrel{w}{\rightarrow} \mathcal{L}^H ( \cdot ) \;\;\; as \; N \rightarrow \infty
\end{eqnarray}
\noindent where $\mathcal{G}^H ( \cdot )$, $\mathcal{L}^H ( \cdot )$ are Gaussian processes defined as in Proposition $\ref{asympt_hadam}$.
\end{corollary}

Corollary $\ref{asymp_resampled_lorenz}$ provides  an
approximation scheme for the probability laws of $\widehat{G}_H (\cdot )$, $\widehat{L}_H (\cdot )$.

As a by-product, a similar result for Gini concentration ratio can be obtained. Let
\begin{eqnarray}
\widehat{R}^*_H = 1 -2 \int_{0}^{1} \widehat{L}^{*}_{H} (p) \, dp
\label{eq:res_r}
\end{eqnarray}
\noindent be the resampled version of the Gini concentration ratio.
\begin{corollary}
\label{asymp_resampled_R}
Under the assumptions of Propositions $\ref{asympt_hadam}$, $\ref{asympt_ricamp_hajek}$, $\sqrt{n} ( \widehat{R}^*_H - \widehat{R}_H )$, as $n$, $N$ go to infinity, tends in distribution to a normal variate with zero mean and variance $\sigma^2_{R,H}$  $( \ref{eq:variance_R} )$.
\end{corollary}

\section{Construction of fixed-size confidence bands}
\label{sec:confidence_band}

Corollary  $\ref{asymp_resampled_lorenz}$ essentially offers a resampling scheme enabling one to approximate the actual (design-based)
distribution of the estimator ${L}^H ( \cdot )$ of the finite population Lorenz curve.
In particular, we focus here on the construction of a confidence band for the superpopulation Lorenz curve,
$L ( \cdot )$. For the sake of simplicity, we confine ourselves on fixed-size confidence bands. Let $d_{1- \alpha}$ be the $(1- \alpha)$-quantile
of the distribution of
\begin{eqnarray}
\sup_{0 \leq p \leq 1} \left \vert \mathcal{L} (p) \right \vert .
\label{eq:sup_abs}
\end{eqnarray}
\noindent If the covariance kernel of the Gaussian process $\mathcal{L} ( \cdot )$ is non-singular, then the r.v. $( \ref{eq:sup_abs} )$ is absolutely continuous with
strictly positive density (cfr. \cite{lifs82}), so that the equality
\begin{eqnarray}
Pr \left ( \sup_{0 \leq p \leq 1} \left \vert \mathcal{L} (p) \right \vert \leq d_{1- \alpha} \right ) = 1- \alpha .
\label{eq:equal_quant}
\end{eqnarray}
\noindent holds.
Relationship $( \ref{eq:equal_quant} )$, in its turn, implies that the region
\begin{eqnarray}
\left [ \max \left ( 0, \,  \widehat{L}_H (p) - \frac{d_{1- \alpha}}{\sqrt{n}} \right ) , \;
\min \left ( 1, \, \widehat{L}_H (p) + \frac{d_{1- \alpha}}{\sqrt{n}} \right ) ; \;\;
0 \leq p \leq 1 \right ]
\label{eq:band_1}
\end{eqnarray}
\noindent is a confidence band for the whole Lorenz curve $L_N ( \cdot )$ of asymptotic level $1- \alpha$.

The quantile $d_{1- \alpha}$ appearing in $( \ref{eq:band_1} )$ obviously depends on the law of the process $\mathcal{L} ( \cdot )$,
and hence cannot be computed in practice. The resampling scheme of Section \ref{sec:resamp} offers the following, simple procedure for its
approximate evaluation.

\begin{itemize}
\item[1] Generate $M$ independent samples of size $n$ on the basis of the two-phase procedure described above.
\item[2] For each generated sample,  compute the corresponding Hajek estimator $( \ref{eq:hajek_resampled} )$. They will be denoted by $\widehat{F}^{*}_{H,m} (y)$, $m=1, \, \dots , \, M$.
\item[3]Compute the corresponding estimates of the Lorenz curve:
\begin{eqnarray}
\widehat{L}^{*}_{m} (p) = \widehat{G}^{*}_{m} (p) / \widehat{G}^{*}_{m} (1) , \;\; 0 \leq p \leq 1 ;
 \;\; m=1, \, \dots , \, M.
\nonumber
\end{eqnarray}
\noindent with
\begin{eqnarray}
\widehat{G}^{*}_{m} (p) = \int_{0}^{p} \inf \{ y: \widehat{F}^{*}_{H,m} (y) \geq u \} \, du .
\nonumber
\end{eqnarray}
\item[4] Compute the $M$ quantities
\begin{eqnarray}
Z^{*}_{n,m} = \sqrt{n} \sup_{0 \leq p \leq 1} \left \vert \widehat{L}^{*}_{m} (p) - \widehat{L}_{H} (p) \right \vert
; \;\;
m=1, \, \dots , \, M .
\label{eq:def_z}
\end{eqnarray}
\end{itemize}

Denote further by
\begin{eqnarray}
\widehat{T}^{*}_{n,M} (z) = \frac{1}{M} \sum_{m=1}^{M} I_{( Z^{*}_{n,m} \leq z )} , \;\; z \in \mathbb{R}
\label{eq:resampling_edf}
\end{eqnarray}
\noindent the empirical distribution function of $Z^{*}_{n,m}$s, and by
\begin{eqnarray}
\widehat{T}^{* -1}_{n,M} ( u) = \inf \{ z : \; \widehat{T}^{*}_{n,M} (z) \geq u \} , \;\; 0 <u<1
\label{eq:resampling_quantiles}
\end{eqnarray}
\noindent the corresponding $u$-quantile.

The empirical d.f. $( \ref{eq:resampling_edf} )$ is essentially an approximation of the (resampling)
distribution of $( \ref{eq:sup_abs} )$.

In Proposition  $\ref{prop_resampling_comp} $
it is stated that $( \ref{eq:resampling_edf} )$
converges to the d.f. of $( \ref{eq:sup_abs} )$, and
that a similar result holds for the quantiles $( \ref{eq:resampling_quantiles} )$. Proof is in Appendix.

\begin{proposition}
\label{prop_resampling_comp}
For almost all $y_i$s, $x_i$s values, and in probability w.r.t. $\DN$, $( N^*_1 , \, \dots , \, N^*_N )$,
conditionally on $\yN$, $\xN$, $\DN$, $( N^*_1 , \, \dots , \, N^*_N )$,
the following results hold:
\begin{eqnarray}
\, & \, & \sup_{z} \left \vert \widehat{T}^{*}_{n,M} (z) - Pr \left ( \sup_{0 \leq p \leq 1} \left \vert \mathcal{L} (p) \right \vert
\leq z \right ) \right \vert \stackrel{a.s.}{\rightarrow} 0 \,
; \label{eq:risult_1} \\
\, & \, & \widehat{T}^{*-1}_{n,M} (u) \stackrel{a.s.}{\rightarrow} d_u , \;\; \forall \,
0<p<1 \label{eq:risult_2}
\end{eqnarray}
\noindent as $M$, $N$ go to infinity.
\end{proposition}

As a consequence of Proposition $\ref{prop_resampling_comp}$, the region
\begin{eqnarray}
\, & \, &
\left [ \widehat{L}_H (p) - n^{-1/2} T^{* -1}_{n,M} (1- \alpha ) , \;
 \widehat{L}_H (p) +  n^{-1/2} T^{* -1}_{n,M} (1- \alpha ) ; \;\; 0 \leq p \leq 1
\right ] \label{conf_band_percentile}
\end{eqnarray}
\noindent is a confidence band for $L ( \cdot )$ with asymptotic level $1- \alpha$ as $N$ and $M$ increase.

\section{Gini concentration index}
\label{sec:conc_index}

The results of the above Section also allow us to construct confidence intervals of asymptotic level $1- \alpha$ for the Gini concentration index $R$. Using the
same notation as in Section $\ref{sec:confidence_band}$, let $\widehat{R}^{*}_{H, m}$, $m=1, \, \dots , M$ the $M$ replicates of $\widehat{R}_H$ generated according to the resampling procedure of Section $\ref{sec:confidence_band}$, and let
\begin{eqnarray}
Z^*_{n,m} = \sqrt{n} ( \widehat{R}^{*}_{H, m} - \widehat{R}_{H} ) , \;\; m=1, \, \dots , \, M.
\nonumber
\end{eqnarray}
Let further
\begin{eqnarray}
\widehat{T}^{*}_{n,M} (z) & = & \frac{1}{M} \sum_{m=1}^{M} I_{( Z^{*}_{n,m} \leq z )} , \;\; z \in \mathbb{R}
\label{eq:resampling_R} \\
\widehat{T}^{* -1}_{n,M} ( u) & = & \inf \{ z : \; \widehat{T}^{*}_{n,M} (z) \geq u \} , \;\; 0 <u<1 \\
\label{eq:resampling_quantiles_R}
\widehat{V}^{*} & = & \frac{1}{M-1} \sum_{m=1}^{M} \left ( Z^{*}_{n,m} - \overline{Z}^{*}_M  \right )^2
\label{eq:resamp_var}
\end{eqnarray}
\noindent
where
\begin{eqnarray}
\overline{Z}^{*}_M  = \frac{1}{M} \sum_{m=1}^{M } Z^{*}_{n,m} . \nonumber
\end{eqnarray}

Denote now by $\Phi_{\mu, \sigma^{2}} ( \cdot )$ the d.f. of a normal distribution with expectation $\mu$ and variance $\sigma^2$.
\begin{proposition}
\label{prop_rasampling_R}
Under the assumptions of Proposition
$\ref{prop_resampling_comp}$, the following results hold:
\begin{eqnarray}
\, & \, & \sup_{z} \left \vert \widehat{T}^{*}_{n,M} (z) - \Phi_{0, \sigma^2_{R,H}} (z) \right \vert \stackrel{a.s.}{\rightarrow} 0 \,
; \label{eq:risult_R1} \\
\, & \, & \widehat{T}^{*-1}_{n,M} (p) \stackrel{a.s.}{\rightarrow} \Phi_{0, \sigma^{2}_{R,H}}^{-1} (p) , \;\; \forall \,
0<p<1 \label{eq:risult_R2}
\end{eqnarray}
\noindent as $M$, $N$ go to infinity.
If, in addition,  the sequence $\left ( Z^{*}_{m} - \overline{Z}^{*}_M  \right )^2 $ is dominated by a r.v. $U$  with finite expectation, then
\begin{eqnarray}
\widehat{S}^{2*} \rightarrow \sigma^2_{R,H} \;\; {\mathrm{as}} \; M, \, N, \, n  \rightarrow \infty
\label{eq:risult_3R}
\end{eqnarray}
\noindent where convergence in $( \ref{eq:risult_3R} )$ is in probability w.r.t. resampling
replications.
\end{proposition}

The main  consequences of Proposition $\ref{prop_rasampling_R}$ are two. First of all, {\em the estimator $\widehat{S}^{2*} $
is a consistent estimator of the variance of $ \widehat{R}_H $}; variance estimation based on linearization techniques is dealt with, for instance, in \cite{barperri16}. In the second place, {\em the confidence intervals}
\begin{eqnarray}
\, & \, &
\left [ \widehat{R}_H - n^{-1/2} T^{* -1}_{n,M} (1- \alpha /2 ) , \;
\widehat{R}_H - n^{-1/2} T^{* -1}_{n,M} ( \alpha /2 )
\right ] \label{conf_int_percentile} \\
\, & \, &
\left [ \widehat{R}_H - n^{-1/2}  z_{\alpha /2} \widehat{S}^{*} , \;
\widehat{R}_H + n^{-1/2}  z_{\alpha /2} \widehat{S}^{*} \right ]
\label{eq:conf_int_norm}
\end{eqnarray}
\noindent {\em both possess asymptotic  confidence level $1- \alpha$} as $N$, $n$ and $M$ increase.

\section{Testing for Lorenz dominance}
\label{sec:lorenz_dominance}

Consider two finite populations  ${\mathcal U}_{N_{h}}$ of size $N_h$, $h=1, \, 2$. In the sequel, we will essentially use the same notation as in Section
$\ref{sec:problem}$, with the addition of the suffix $h$ ($=1, \, 2$).
Denote by
$$
L_{N_{h}} ( p) , \;\; 0 \leq p \leq 1
$$
the Lorenz curve for population ${\mathcal U}_{N_{h}}$, $h=1$, $2$.
The population ${\mathcal U}_{N_{1}}$ (weakly) {\em Lorenz dominates}  ${\mathcal U}_{N_{2}}$
if
\begin{eqnarray}
L_{2} (p) \leq L_{1} (p) \;\;  \forall \, 0 \leq p \leq 1 \nonumber
\end{eqnarray}
\noindent {\em i.e.} if
\begin{eqnarray}
\phi (p) \geq 0 \;\;  \forall \, 0 \leq p \leq 1 \label{eq:lorenz_dom}
\end{eqnarray}
\noindent where
\begin{eqnarray}
\phi (p) = L_{1} (p) - L_{2} (p) . \nonumber
\end{eqnarray}
The goal of the present section is to construct a test for the Lorenz dominance hypothesis
\begin{itemize}
\item[$H_{0}$:] $\; \phi (p) \geq 0 \;\; \forall \, 0 \leq p \leq 1.$
\item[$H_{1}$:] $\; \phi (p) < 0 \;\; {\mathrm{for \; some \;}} p \in [0, \, 1]$.
\end{itemize}
The importance of the Lorenz ordering, and  testing for Lorenz dominance is stressed, for instance, in \cite{anderson96}, \cite{barrett14}.

From population ${\mathcal U}_{N_{h}}$ a sample of size $n_h$ is drawn, according to a sampling design $P_h$ satisfying assumptions $A4$-$A6$. Assume further that all r.v.s $Y_{h,i}$s ($i=1, \, \dots , \, N_h$, $h=1, \, 2$) for the two populations are independent, and that the two sampling designs $P_1$, $P_2$ independently select samples from ${\mathcal U}_{N_{1}}$, ${\mathcal U}_{N_{2}}$. The following proposition is an immediate consequence of Proposition $\ref{asympt_hadam}$.

\begin{proposition}
\label{asympt_dom1}
Suppose assumptions A1-A6, C1, D1 hold for ${\mathcal U}_{N_{h}}$, $h=1, \, 2$, and that $n_1 / (n_1 + n_2 ) \rightarrow \tau$ as $N_1$, $N_2$ go to infinity, with $0 < \tau < 1$. Let further $\widehat{\phi}_H (p) = \widehat{L}_{H,1} (p) - \widehat{L}_{H,2} (p)$, and
\begin{eqnarray}
\mathcal{F}^H_{N_{1}, N_{2}} ( \cdot ) = \sqrt{\frac{n_{1} n_{2}}{n_{1} + n_{2}}} ( \widehat{\phi}_H ( \cdot ) - \phi ( \cdot ) ) . \label{eq:lordiff_proc}
\end{eqnarray}
Then, as $N_1$, $N_2$ tend to infinity, the sequence of stochastic  processes $( \ref{eq:lordiff_proc} )$ converges weakly to a Gaussian process that can be represented as
\begin{eqnarray}
\mathcal{F}^H ( \cdot ) = \sqrt{1- \tau} \mathcal{L}^H_1 ( \cdot ) + \sqrt{\tau } \mathcal{L}^H_2 ( \cdot ) \nonumber
\end{eqnarray}
\noindent where $\mathcal{L}^H_h ( \cdot )$, $h=1$, $2$ are two independent Gaussian processes having representation $( \ref{eq:proc_l} )$.
\end{proposition}

The idea pursued here to construct a test procedure for the hypothesis of Lorenz dominance is simple. It is summarized below.
\begin{itemize}
\item[-] Construct a confidence band of level $1- \alpha$ for $ \phi ( \cdot ) $.
\item[-] If, for at least a $p \in (0, \, 1)$ the confidence band is under the horizontal axis, reject the stochastic dominance hypothesis.
\item[-]Otherwise, ``accept'' stochastic dominance hypothesis.
\end{itemize}
Clearly, the test procedure has significance level equal to $\alpha$. A graphical illustration of the procedure is in Fig. $\ref{fig:figura1}$ A, B.

\begin{figure}
\includegraphics[width=0.9\textwidth]{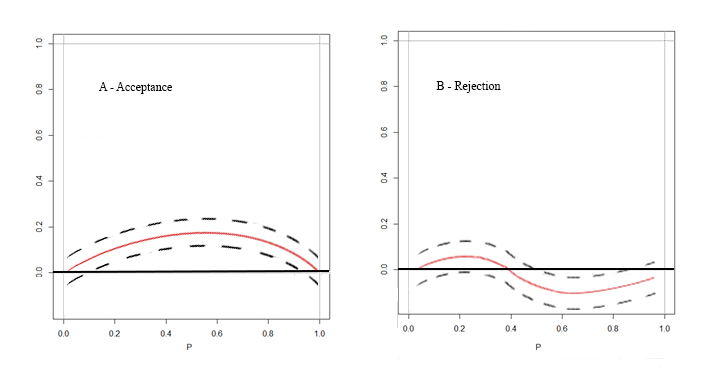}
\caption{Graphical illustration of the test procedure}
\label{fig:figura1}
\end{figure}

A confidence band for $\phi ( \cdot )$ can be constructed by using the resampling procedure of Section $\ref{sec:resamp}$.

\begin{itemize}
\item[1.] For sample drawn from population ${\mathcal U}_{N_{h}}$, generate $M$ independent samples of size $n_h$ $(h=1, \, 2)$ on the basis of the two-phase procedure described in Section $\ref{sec:resamp}$.
\item[2.] For each generated sample, compute the corresponding estimates of the Lorenz curves:
\begin{eqnarray}
\widehat{L}^{*}_{h,m} (p) , \;\; 0 \leq p \leq 1 ;
 \;\; m=1, \, \dots , \, M, \; h=1, \, 2.
\nonumber
\end{eqnarray}
\item[3.] Compute the $M$ quantities
\begin{eqnarray}
\widehat{\phi}^{*}_{m} (p) =
\widehat{L}^{*}_{1,m} (p) -
\widehat{L}^{*}_{2,m} (p) , \;\; 0 \leq p \leq 1 ;
 \;\; m=1, \, \dots , \, M.
\nonumber
\end{eqnarray}
\item[4.] Compute the $M$ quantities
\begin{eqnarray}
Z^{*}_{1,2,m} = \sqrt{\frac{n_{1} n_{2}}{n_{1} + n_{2}}} \sup_{0 \leq p \leq 1} \left \vert \widehat{\phi}^{*}_{m} (p) - \widehat{\phi} (p) \right \vert
; \;\;
m=1, \, \dots , \, M .
\label{eq:def_zdiff}
\end{eqnarray}
\end{itemize}

Denote now by
\begin{eqnarray}
\widehat{T}^{*}_{1,2,M} (z) = \frac{1}{M} \sum_{m=1}^{M} I_{( Z^{*}_{1,2,,m} \leq z )} , \;\; z \in \mathbb{R}
\label{eq:resampling_edfdiff}
\end{eqnarray}
\noindent the empirical distribution function of $Z^{*}_{1,2,m}$s, and by
\begin{eqnarray}
\widehat{T}^{* -1}_{1,2,M} ( u) = \inf \{ z : \; \widehat{T}^{*}_{1,2,M} (z) \geq u \} , \;\; 0 <u<1
\label{eq:resampling_quantiles_diff}
\end{eqnarray}
\noindent the corresponding $u$-quantile.

Using the same arguments as in Proposition $\ref{prop_resampling_comp}$, it is now possible to prove the following, further result.
\begin{proposition}
\label{prop_resampling_compdiff}
For almost all $y_{h, i}$s, $x_{h, i}$s values $(h=1, \, 2)$, and in probability w.r.t. the sample designs,
the following results hold:
\begin{eqnarray}
\, & \, & \sup_{z} \left \vert \widehat{T}^{*}_{1,2,M} (z) - Pr \left ( \sup_{0 \leq p \leq 1} \left \vert \mathcal{F} (p) \right \vert
\leq z \right ) \right \vert \stackrel{a.s.}{\rightarrow} 0 \,
; \label{eq:risult_1diff} \\
\, & \, & \widehat{T}^{*-1}_{1,2,M} (u) \stackrel{a.s.}{\rightarrow} d_u , \;\; \forall \,
0<p<1 \label{eq:risult_2diff}
\end{eqnarray}
\noindent as $M$, $n_h$, $N_h$ $(h=1, \, 2)$ go to infinity.
\end{proposition}

As a consequence of Proposition $\ref{prop_resampling_compdiff}$, the region
\begin{eqnarray}
\, & \, &
\left [ \widehat{\phi}_H (p) - \sqrt{\frac{n_{1} + n_{2}}{n_{1} n_{2}}} T^{* -1}_{1,2,M} (1- \alpha ) , \;
 \widehat{\phi}_H (p) +  \sqrt{\frac{n_{1} + n_{2}}{n_{1} n_{2}}} T^{* -1}_{1,2,M} (1- \alpha ) ; \;\; 0 \leq p \leq 1
\right ] \label{conf_band_percentile_diff}
\end{eqnarray}
\noindent is a confidence band for $\phi ( \cdot )$ with asymptotic level $1- \alpha$ as $N_h$ $(h=1, \, 2)$ and $M$ increase.

\section{Simulation study}
\label{sec:simulation}

In this section a simulation study is performed, in order to evaluate the actual confidence level of the proposed confidence bands and intervals.
The  simulation scenario is similar to \cite{antal:11}. In detail,  a finite population of size $N$ has been generated from the model
\begin{eqnarray}
y_i=(\beta_{0}+\beta_1x_{i}^{1.2}+\sigma\epsilon_i)^2+c
\end{eqnarray}
where $x_i=|j_i|$ and $j_i \sim N(0,7)$, $\epsilon_i \sim N(0,1)$ and $\sigma=15$. According to \cite{antal:11}, the regression parameters $\beta_0=12.5$, $\beta_1=3$ and $c=4000$ have been chosen.
As far as the inclusion probabilities are concerned, they are taken proportional to the value of a variable $Z$, generated from the equation $Z=Y^{0.2} W$ where $W$ has a lognormal distribution $(ln N(\mu,\sigma^2))$ with parameters $\mu=0$ and $\sigma^2=0.025$.
Three population sizes, $N= 300$, $500$, $1000$, and one sampling fraction ($n/N = 0.2$) have been considered. For each combination $(n, \, N)$, $1000$ samples of size $n$ have been generated according to two sampling schemes: Pareto design (PA) and Sampford design (SA), with inclusion probabilities proportional to $z_i$s.
Actual coverage probabilities for confidence bands for the superpopulation Lorenz curve $L( \cdot )$, as well as for confidence intervals for the superpopulation Gini concentration ratio $G$, have been computed. The nominal level was $0.95$ in all cases. Results are shown in Table $\ref{tab:tab01}$.
\bigskip

\begin{table}[htbp]
\caption{Coverage probabilities and interval size (in parenthesis) for Lorenz curve and Gini concentration ratio. Nominal level: $0.95$}
\label{tab:tab01}
\begin{tabular}{|c|c|c||c|c|c|}
\hline
\multicolumn{3}{|c||}{Pareto sampling design} & \multicolumn{3}{|c|}{Sampford sampling design} \\
\hline \hline
\multicolumn{6}{|c|}{Population ($N$) and sample ($n$) sizes}\\
\hline
$N=300$ & $N=500$ & $N=1000$ & $N=300$ & $N=500$ & $N=1000$ \\
($n=60$) & ($n=100$) & ($n=200$) & ($n=60$) & ($n=100$) & ($n=200$)\\
\hline \hline
\multicolumn{6}{|c|}{Confidence band for Lorenz curve}\\
\hline
$0.892$ &	$0.918$ & $0.947$ & $0.889$ & $0.919$ & $0.951$ \\
($0.086$) & ($0.068$) & ($0.050$) & ($0.086$) & ($0069$) & ($0.050$) \\
\hline \hline
\multicolumn{6}{|c|}{Confidence interval for Gini coefficient (Normal approximation)}\\
\hline
$0.875$ & $0.915$ & $0.944$ & $0.877$ & $0.915$ & $0.946$ \\
($0.080$) & ($0.065$) & ($0.047$) & ($0.081$) & ($0.065$) & ($0.047$) \\
\hline \hline
\multicolumn{6}{|c|}{Confidence interval for Gini coefficient (Pivot percentile method)}\\
\hline
$0.861$ & $0.904$ & $0.922$ & $0.863$ & $0.907$ & $0.926$ \\
($0.078$) & ($0.064$) & ($0.046$) & ($0.080$) & ($0.065$) & ($0.047$) \\
\hline \hline
\end{tabular}
\end{table}

As it appear from Table $\ref{tab:tab01}$, the performance of the confidence band for the whole Lorenz curve is generally good, and the actual coverage
probability is close to the nominal level $0.95$, for both Pareto and Sampford sampling designs.
Similar considerations hold for confidence intervals for Gini concentration ratio, $G$. The method based on normal approximation and variance estimated by
resampling $( \ref{eq:conf_int_norm} )$ performs slightly better than the pivot-percentile method
$( \ref{conf_int_percentile} )$. Again, results are virtually identical for both Pareto and Sampford designs.

\newpage

{\Large \noindent {\bf Appendix}}

\noindent
\begin{proof}[{\bf Proof of Proposition \ref{continuous_traject}}]
Suppose that $y < t$. From $( \ref{eq:cov_ker_hajek2} )$ it is not difficult to see that
\begin{eqnarray}
E[ ( \mathcal{W}^H (t) - \mathcal{W}^H (y) )^2 ] & = & C_1 (y, \, y) + C_1 (t, \, t) - 2 C_1 (t, \, y)
+ f ( C_2 (y, \, y) + C_2 (t, \, t) - 2 C_2 (t, \, y) )
 \nonumber \\
& = & \mathbb{E} [ X_1 ] \left ( T_{-1} (t) - T_{-1} (y) \right ) + f \left ( F(t) - F(y) \right ) \nonumber \\
\, & \, & - \frac{f^{3}}{d} \left \{ \left ( \frac{T_{1} (y)}{\mathbb{E} [ X_{1} ]} - F(y)  \right )
\left ( \frac{T_{1} (t)}{\mathbb{E} [ X_{1} ]} - F(t)  \right )
\right \} ^2 \nonumber \\
\, & \, & - 2 \mathbb{E} [ X_1 ] ( T_{-1} (y) - T_{-1} (y) ) ( F(y) - F(t)) \nonumber \\
\, & \, & + 2 \mathbb{E} [ X_1 ] \left ( \mathbb{E} [ X_1^{-1} ] +1 \right ) ( F(t) - F( y) )^2 \nonumber \\
\, & \, & + f \{ (F(t) - F(y) )  - (F(t) - F(y) )^2 \} .
\label{eq:disug_kolm_1}
\end{eqnarray}
Assumption C1 implies that
\begin{eqnarray}
\left \vert T_{\alpha} (t) -  T_{\alpha} (y) \right \vert \leq M_{\alpha} \left \vert F (t) - F(y) \right \vert \nonumber
\end{eqnarray}
\noindent so that from $( \ref{eq:disug_kolm_1} )$ it is not difficult to see that
\begin{eqnarray}
E[ ( \mathcal{W}^H (t) - \mathcal{W}^H (y) )^2 ] \leq C \left \vert F(t) - F(y) \right \vert .
\label{eq:disug_kolm_2}
\end{eqnarray}
\noindent $C$ being an appropriate constant. Inequality $( \ref{eq:disug_kolm_2} )$ also holds when $y>t$. Hence, in terms of the process
$B^H$ introduced above we may write
\begin{eqnarray}
\mathbb{E} \left [ \left ( B^H (t) - B^H (y) \right )^2 \right ] \leq C \vert t-y \vert \;\; \forall \, y, \, t \in [0, \, 1] .
\label{eq:disug_kolm_3}
\end{eqnarray}
Inequality $( \ref{eq:disug_kolm_3} )$ and the Gaussianity of $B^H (t) - B^H (y)$, in their turn, imply that
\begin{eqnarray}
\mathbb{E} \left [ \left ( B^H (t) - B^H (y) \right )^2 \right ] \leq \frac{C}{\left \vert \log ( t-y ) \right \vert^{\beta}} \;\;
 \forall \, \beta > 1 , \; \forall \, y, \, t \in [0, \, 1] .
\label{eq:disug_kolm_4}
\end{eqnarray}
Observing that $Pr ( B^H (0) =0) = Pr (B^H (1) =0) =1$,
Proposition $\ref{continuous_traject}$ now follows from $( \ref{eq:disug_kolm_4} )$ and \cite{leadweis69}.
\end{proof}

\noindent
\begin{proof}[{\bf Proof of Proposition \ref{prop_resampling_comp}}]
Let
$$
R^{*}_{n} (z) = Pr_{P^{*}} \left ( \left .  Z^{*}_{n,m} \leq z  \, \right
 \vert \yN , \, \xN , \, \DN , \, N^*_1 , \, \dots N^*_M  \right )
$$
be the (resampling) d.f. of
$Z^{*}_{n,m} $ $( \ref{eq:def_z} )$. By Dvoretzky-Kiefer-Wolfowitz inequality (cfr. \cite{massart90}), we have first
\begin{eqnarray}
Pr \left ( \left . \sup_{z} \left \vert \widehat{R}^{*}_{n,M} (z) - R^{*}_{n} (z)  \right \vert  > \epsilon \,
\right \vert \yN , \, \xN , \, \DN , \, N^*_1 , \, \dots N^*_M \right ) \leq 2 \exp \left \{ -2 M \epsilon^2   \right \} .
\label{eq:dkw_ineq}
\end{eqnarray}
\noindent Using the Borel-Cantelli first lemma, and taking into account that $R^{*}_{n} (z)$ converges uniformly to
$Pr ( \sup_p \vert  \mathcal{L} (p) \leq z )$, $( \ref{eq:risult_1} )$ immediately follows. Statement $( \ref{eq:risult_2} )$
follows from $( \ref{eq:risult_1} )$ and the absolute continuity of the distribution of $\sup_p \vert \mathcal{L} (p) \vert $
(cfr. \cite{lifs82}).
\end{proof}

\begin{proof}[{\bf Proof of Proposition \ref{prop_rasampling_R}}]
Proof of $( \ref{eq:risult_R1} )$ and $( \ref{eq:risult_R2} )$ is similar to Proposition $\ref{prop_resampling_comp}$.
As far as $( \ref{eq:risult_3R} )$ is concerned, it is a consequence of Th. 2.5.5. in \cite{sensinger:93} (pp. 90-91).
\end{proof}

\newpage

\bibliography{sampling_nov_13_2018}

\begin{thebibliography}{}

\bibitem[Anderson(1996)]{anderson96}
Anderson, G. (1996).
\newblock Nonparametric {T}ests of {S}tochastic {D}ominance in {I}ncome
  {D}istribution.
\newblock {\em Econometrica}, {\bf 64}, 1183--1193.

\bibitem[Antal and Till\'{e}(2011)]{antal:11}
Antal, E. and Till\'{e}, Y. (2011).
\newblock A direct bootstrap method for complex sampling designs from a finite
  population.
\newblock {\em Journal of the American Statistical Association}, {\bf
  106}(494), 534--543.

\bibitem[Barabesi {\em et~al.}(2016)]{barperri16}
Barabesi, L., Diana, G., and Perri, P.~F. (2016).
\newblock Linearization of inequality indices in the design-based framework.
\newblock {\em Statistics}, {\bf 50}, 1161--1172.

\bibitem[Barrett {\em et~al.}(2014)]{barrett14}
Barrett, G.~F., Donald, S.~G., and Bhattacharya, D. (2014).
\newblock Consistent {N}onparametric {T}ests for {L}orenz {D}ominance.
\newblock {\em Journal of Business and Economic Statistics}, {\bf 32}, 1--13.

\bibitem[Bhattacharya(2005)]{bhattacharya05}
Bhattacharya, D. (2005).
\newblock Asymptotic inference from multi-stage samples.
\newblock {\em Journal of Econometrics}, {\bf 126}, 145--171.

\bibitem[Bhattacharya(2007)]{bhattacharya07}
Bhattacharya, D. (2007).
\newblock Inference on inequality from household survey data.
\newblock {\em Journal of Econometrics}, {\bf 137}, 674--707.

\bibitem[Bickel and Freedman(1981)]{bickelfried81}
Bickel, P.~J. and Freedman, D. (1981).
\newblock Some asymptotic theory for the bootstrap.
\newblock {\em The Annals of Statistics}, {\bf 9}, 1196--1216.

\bibitem[Boistard {\em et~al.}(2017)]{bastard17}
Boistard, H., Lopuha\"{a}, R., and Ruiz-Gazen, A. (2017).
\newblock Functional central limit theorems for single-stage sampling designs.
\newblock {\em The Annals of Statistics}, {\bf 45}, 1728--1758.

\bibitem[Chauvet(2007)]{chauvet07}
Chauvet, G. (2007).
\newblock {\em M\'{e}thodes de bootstrap en population finie}.
\newblock Ph.D. Dissertation, Laboratoire de statistique d'enqu\^{e}tes,
  CREST-ENSAI, Universiot\'{e} de Rennes 2.

\bibitem[Conti and {Di Iorio}(2018)]{contiorio18}
Conti, P.~L. and {Di Iorio}, A. (2018).
\newblock Analytic inference in finite populations via resampling, with
  applications to confidence intervals and testing for independence.
\newblock {\em Preprint arXiv:1809.08035 available at
  https://arxiv.org/abs/1809.08035 - {\rm{Submitted for publication}}}.

\bibitem[Conti and Marella(2015)]{contimar14}
Conti, P.~L. and Marella, D. (2015).
\newblock Inference for quantiles of a finite population: asymptotic vs.
  resampling results.
\newblock {\em Scandinavian Journal of Statistics}, {\bf 42}, 545--561.

\bibitem[Cs\"{o}rg\H{o} {\em et~al.}(1986)]{csorgo86}
Cs\"{o}rg\H{o}, M., Cs\"{o}rg\H{o}, S., and Horv\'{a}th, L. (1986).
\newblock {\em An Asymptotic Theory for Empirical Reliability and Concentration
  Processes}.
\newblock Springer-Verlag, Berlin.

\bibitem[Davidson(2009)]{davidson09}
Davidson, R. (2009).
\newblock Reliable inference for the {G}ini index.
\newblock {\em Journal of Econometrics}, {\bf 150}, 30--40.

\bibitem[Gastwirth(1972)]{gastwirth72}
Gastwirth, J.~L. (1972).
\newblock The estimation of {L}orenz curve and {G}ini index.
\newblock {\em Review of Economics and Statistics}, {\bf 54}, 306--316.

\bibitem[Giorgi(1999)]{giorgi:99}
Giorgi, G.~M. (1999).
\newblock {\em {\rm{Income {I}nequality {M}easurement: {T}he {S}tatistical
  {A}pproach. {I}n:}} {H}anbdbook of {I}ncome {I}nequtality {M}easurement
  {\rm{(J. Silber Ed.).}}}
\newblock Kluwer Academic Publishers, Boston.

\bibitem[Giorgi and Gigliarano(2017)]{giogig17}
Giorgi, G.~M. and Gigliarano, C. (2017).
\newblock The {G}ini {C}oncentration {I}ndex: a {R}eview of the {I}nference
  {L}iterature.
\newblock {\em Journal of Economic Surveys}, {\bf 31}, 1130--1148.

\bibitem[Goldie(1977)]{goldie77}
Goldie, C.~M. (1977).
\newblock Convergence theorems for empirical {L}orenz curve and their inverses.
\newblock {\em The Annals of Applied Probability}, {\bf 9}, 765--791.

\bibitem[H\'{a}jek(1964)]{hajek64}
H\'{a}jek, J. (1964).
\newblock Asymptotic theory of rejective sampling with varying probabilities
  from a finite population.
\newblock {\em The Annals of Mathematical Statistics}, {\bf 35}, 1491--1523.

\bibitem[Langel and Till\'{e}(2013)]{tillelangel13}
Langel, M. and Till\'{e}, Y. (2013).
\newblock Variance estimation of the {G}ini index: revisiting a result several
  times published.
\newblock {\em Journal of the Royal Statistical Society {\rm{Series A}}}, {\bf
  176}, 521�540.

\bibitem[Leadbetter and Weissner(1969)]{leadweis69}
Leadbetter, M.~R. and Weissner, J.~H. (1969).
\newblock On continuity and other analytic properties of stochastic process
  sample functions.
\newblock {\em Proceedings of the American Mathematical Society}, {\bf 22},
  291--294.

\bibitem[Lifshits(1982)]{lifs82}
Lifshits, M.~A. (1982).
\newblock On the absolute continuity of distributions of functionals of random
  processes.
\newblock {\em Theory of Probability and Its Applications}, {\bf 27}, 600--607.

\bibitem[Marella and Vicard(2018)]{marvicard18}
Marella, D. and Vicard, P. (2018).
\newblock Pc complex: {PC} algorithm for complex survey data.
\newblock {\em Working Paper n. 240, Dipartimento di Economia - Universit\`{a}
  Roma Tre (ISSN: 2279-6916). {\rm{Submitted for publication}}}.

\bibitem[Massart(1990)]{massart90}
Massart, P. (1990).
\newblock The tight constant in the {D}voretzky-{K}iefer-{W}olfowitz
  inequality.
\newblock {\em The Annals of Probability}, {\bf 18}, 1269--1283.

\bibitem[Pfeffermann(1993)]{pfeffermann93}
Pfeffermann, D. (1993).
\newblock The role of sampling weights when modeling survey data.
\newblock {\em International Statistical Review}, {\bf 61}, 317--337.

\bibitem[Pfeffermann and Sverchkov(2004)]{pfeffersver04}
Pfeffermann, D. and Sverchkov, M. (2004).
\newblock Prediction of finite population totals based on the sample
  distribution.
\newblock {\em Survey Methodology}, {\bf 30}, 79--92.

\bibitem[Sen and Singer(1993)]{sensinger:93}
Sen, P.~K. and Singer, J. (1993).
\newblock {\em Large Sample Methids in Statistics}.
\newblock Champam \& Hall, London.

\bibitem[Till\'{e}(2006)]{tille06}
Till\'{e}, Y. (2006).
\newblock {\em Sampling {A}lgorithms}.
\newblock Springer, New York.

\bibitem[van~der Vaart(1998)]{vandervaart98}
van~der Vaart, A. (1998).
\newblock {\em Asymptotic Statistics}.
\newblock Cambridge University Press, Cambridge.

\bibitem[Zheng(2002)]{zheng02}
Zheng, B. (2002).
\newblock Testing {L}orenz curves with non-simple random samples.
\newblock {\em Econometrica}, {\bf 70}, 1235--1243.

\end{thebibliography}
\bibliographystyle{natbib}

\end{document}